\documentclass[12pt,letterpaper]{article}

\usepackage{times}
\usepackage{amsmath, amsthm, amssymb, amstext, comment, graphicx, fullpage}
\usepackage{xspace}

\usepackage[linesnumbered,ruled,vlined,boxed]{algorithm2e}
\usepackage{amsmath,amssymb}
\usepackage{verbatim}
\usepackage{enumerate}

\newtheorem{definition}{Definition}

\newtheorem{theorem}{Theorem}
\newtheorem{lemma}{Lemma}

\newtheorem{remark}{Remark}

\renewcommand{\vec}{\mathbf}

\usepackage{pifont}% http://ctan.org/pkg/pifont

\newcommand{\cmark}{\ding{51}}%
\newcommand{\xmark}{\ding{55}}%

\title{The Adjusted Winner Procedure: Characterizations and Equilibria}

\author{
\textbf{Haris Aziz}\\
\small{NICTA and University of New South Wales, Australia}\\
\small{\texttt{haris.aziz@nicta.com.au}}
\and
\textbf{Simina Br\^anzei}\\
\small{Hebrew University of Jerusalem, Israel}\\
\small{\texttt{simina.branzei@gmail.com}}\\
\newline
\and
\textbf{Aris Filos-Ratsikas}\\
\small{Oxford University, United Kingdom}\\
\small{\texttt{aris.filos-ratsikas@cs.ox.ac.uk}}
\and
\textbf{S{\o}ren Kristoffer Stiil Frederiksen}\\
\small{Aarhus University, Denmark}\\
\small{\texttt{ssf@cs.au.dk}}
}
\date{}
\begin{document}
\maketitle

\begin{abstract}
The Adjusted Winner procedure is an important fair division mechanism proposed by Brams and Taylor for allocating goods between two parties. It has been used in practice for divorce settlements and analyzing political disputes. Assuming truthful declaration of the valuations, it computes an allocation that is \emph{envy-free}, \emph{equitable} and \emph{Pareto optimal}.

We show that Adjusted Winner admits several elegant characterizations, which further shed light on the outcomes reached with strategic agents. We find that
the procedure may not admit pure Nash equilibria in either the discrete or continuous variants, but is guaranteed to have $\epsilon$-Nash equilibria for each $\epsilon > 0$. Moreover, under \emph{informed} tie-breaking, exact pure Nash equilibria always exist, are Pareto optimal, and their social welfare is at least $3/4$ of the optimal.
\end{abstract}

\thispagestyle{empty}
\setcounter{page}{0}

\newpage

\section{Introduction}
The \emph{Adjusted Winner} procedure was introduced by Brams and Taylor (\cite{BT96b}) as a highly desirable mechanism for allocating 
\emph{multiple divisible resources} among two parties. The procedure requires the participants to declare their preferences over the items and the outcome satisfies strong fairness and efficiency properties.
	 Adjusted Winner has been advocated as a fair division rule for divorce settlements~\cite{BT96b}, international border conflicts~\cite{TaPa08a}, political issues~\cite{DeBr97a,Mass00a}, real estate disputes~\cite{Levy99a}, water disputes~\cite{Mada10a}, deciding debate formats~\cite{Lax99a} and various  negotiation settings~\cite{BrTa00a,Rait00a}. %BrTa96a
	 For example, it has been shown that the agreement reached during Jimmy Carter's presidency between Israel and Egypt is very close to what Adjusted Winner would have predicted~\cite{BrTo96a}.
Adjusted Winner has been patented by New York University and licensed to the law firm Fair Outcomes, Inc~\cite{KKP14a}.

	 Although the merits of Adjusted Winner have been discussed in a large body of literature, the procedure is still not fully understood theoretically. We provide two novel characterizations, together with an alternative interpretation that turns out to be very useful for analyzing the procedure.
	 
	 Moreover, as observed already in \cite{BT96a}, the procedure is susceptible to manipulation. However, fairness and efficiency are only guaranteed when the participants declare their preferences honestly. In a review of a well-known book on Adjusted Winner by Brams and Taylor \cite{BrTa00a}, Nalebuff \cite{Nalebuff} highlights the need for research in this direction:  
	 \begin{quote}
	 	\emph{..thus we have to hypothesize how they (the players) would have played the game and where they would have ended up.}
	 \end{quote}
	 We answer these questions by studying the existence, structure, and properties of pure Nash equilibria of the procedure. Until now, our understanding of the strategic aspects has been limited to the case of two items \cite{BT96a} and experimental predictions \cite{daniel2005fair}; our work identifies conditions under which Nash equilibria exist and provides theoretical guarantees for the performance of the procedure in equilibrium.
	 
	 \subsection{Contributions} 
	 
%and in the case Adjusted Winner, it disallows manipulation by a player by scaling down his bids.
	 
		%An allocation is maxmin if it maximizes the minimum utility of the players. 

% The following theorem is probably proven in one of Brahm's and Taylor's books.
%
% \begin{theorem}
% Assuming truthful reporting, the allocation produced by the Adjusted Winner mechanism is envy-free, equitable and Pareto optimal.
% \end{theorem}

% \simina{Include more references, there are quite a few papers citing this mechanism, some of which in ``real'' contexts, such as the Israeli-Palestinian conflict, or how the agreement reached during Jimmy Carter's presidency between Israel and Egypt is very close to what Adjusted Winner would have predicted}

% \cite{EKH01a,HoLll10a}
%
%
% \haris{Recent AAAI reference on EF allocations: \cite{KKP14a}}
%
%
% \haris{Need to motivate WHY identifying PNE is interesting/important. Something on the following lines. Since Adjusted Winner is not strategyproof, players may best respond in order to obtain a better outcomes. This raises the question about existence, structure, and properties of pure Nash equilibrium under Adjusted Winner}

We start by presenting the first characterizations of Adjusted Winner. We show that among all protocols that split at most one item, it is the only one that satisfies Pareto-efficiency and equitability. Under the same condition, we further show that it is equivalent to the protocol that always outputs a maxmin allocation.
		%In the first characterization, Adjusted Winner is the only Pareto optimal and equitable procedure for two players in which  a maximum of one item needs to be divided partially. 
		%In the second characterization, we show that Adjusted Winner is equivalent to the maxmin rule which %requires a maximum of one item to be divided partially.
		%Finally, we show that if the ratio of utilities for any two items are different, then the only Pareto optimal and equitable allocation is the result of Adjusted Winner.
		
Next, we obtain a complete picture for the existence of pure Nash equilibria in Adjusted Winner. We find the following that neither the discrete nor the continuous variants of the procedure are guaranteed to have pure Nash equilibria. However, the continuous variant of the procedure has $\epsilon$-Nash equilibria, for every $\epsilon>0$, while the discrete variant has $\epsilon$-equilibria when the number of points is chosen appropriately, in a way that allows the players to sketch their valuations precisely enough. Additionally, under \emph{informed} tie-breaking, exact pure Nash equilibria always exist for both variants of the procedure. 
		%We find the following: neither the discrete nor the continuous variants of Adjusted Winner are guaranteed to have pure Nash equilibria. However, in the continuous case, for each $\epsilon > 0$, there exists
		%an $\epsilon$-equilibrium; in the discrete case, for each instance of the valuations and each $\epsilon > 0$, the center can set the number of points large enough such that the protocol has $\epsilon$-equilibria.
		
Finally, we prove that the pure Nash equilibria of Adjusted Winner are envy-free and Pareto optimal with respect to the true valuations and that their social welfare is at least $3/4$ of that of the intended outcome of the procedure.
			%Adjusted Winner for two items and discrete strategies always has a pure Nash equilibrium. 
	%	\haris{Isn't PNE envy-free even if the number of items is not 2??}
	%	Any such PNE is envy-free and Pareto optimal. The Price of Anarchy of the Adjusted Winner procedure is 4/3.
Our results concerning the existence or non-existence of pure Nash equilibria are summarized in Table~\ref{tab:summary}.

\begin{table}
\begin{center}
\begin{tabular}{|c||c|c|}
\hline
%%%\begin{tabular}{@{}c@{}}\small{\emph{Lexicographic}} \\ \small{\emph{tie-breaking}}\end{tabular}& \small{Social Cost}& \small{Max Cost}\\
\begin{tabular}{@{}c@{}}\small{\emph{Continuous}} \\ \small{\emph{Procedure}}\end{tabular} & \begin{tabular}{@{}c@{}}\small{\emph{Lexicographic}} \\ \small{\emph{tie-breaking}}\end{tabular} & \begin{tabular}{@{}c@{}}\small{\emph{Informed}} \\ \small{\emph{tie-breaking}}\end{tabular}\\
\hline\hline
\small{pure Nash} & \xmark & \cmark \\
\hline\hline
\small{$\epsilon$-Nash} & \cmark & \cmark \\
\hline
\end{tabular}
\quad
\begin{tabular}{|c||c|c|}
\hline
\begin{tabular}{@{}c@{}}\small{\emph{~~~~Discrete}} \\ \small{\emph{Procedure}}\end{tabular}& \begin{tabular}{@{}c@{}}\small{\emph{Lexicographic}} \\ \small{\emph{tie-breaking}}\end{tabular}  & \begin{tabular}{@{}c@{}}\small{\emph{Informed}} \\ \small{\emph{tie-breaking}}\end{tabular} \\
\hline\hline
\small{pure Nash} & \xmark & \cmark \\
\hline\hline
\small{$\epsilon$-Nash} &  \cmark $^{(*)}$ & \cmark  \\
\hline
\end{tabular}
\end{center}
\caption{Existence of pure Nash equilibria in Adjusted Winner. The (*) result holds when the number of points is chosen appropriately.}
\label{tab:summary}
\end{table}

% \haris{
% Cite \cite{BCD+14a} in relation to recent equilibria results in AI}
%

\section{Background}
We begin by introducing the classical fair division model for which the Adjusted Winner procedure was developed~\cite{BT96a}.
Let there be two players, Alice and Bob, that are trying to split a set 
$M = \{1, \ldots, m\}$ of divisible items. The players have preferences over the items given by numerical values that express their level of satisfaction.
Formally, let $\mathbf{a}=(a_1,a_2,\ldots,a_m)$ and $\vec{b} = (b_1, \ldots, b_m)$ 
denote their \emph{valuation vectors}, where $a_j$ and $b_j$ are the values assigned by Alice and Bob to item $j$, respectively.

An \emph{allocation} $W = (W_{A}, W_{B})$ is an assignment of fractions of items (or \emph{bundles}) to the players, 
where $W_{A} = (w_{A}^{1},\ldots,w_{A}^{m}) \in [0,1]^m$ and $W_B = (w_{B}^1,\ldots,w_{B}^m) \in [0,1]^m$ are the 
allocations of Alice and Bob, respectively. 

%\begin{definition}[Utility]
The players have additive \emph{utility} over the items. Alice's utility for a bundle $W_{A}$, given that her valuation is $\vec{a}$, is: 
$u_{\mathbf{a}}(W_{A})=\sum_{j \in M} a_j \cdot w_{A}^j$. Bob's utility is defined similarly. 
%\end{definition}
The players are weighted equally, such that their utility for receiving all the resources is the same: $$\sum_{i\in M} a_i = \sum_{i \in M} b_i.$$

There are two main settings studied in this context: \emph{discrete} and \emph{continuous} valuations.
In the discrete setting, valuations are positive natural numbers that add up to some integer $P$ and can be interpreted as \emph{points} (or coins of equal size) that the 
players use to acquire the items. For ease of notation, we will consider the equivalent interpretation of valuations as rationals with common denominator $P$, where the valuations sum to $1$. 
In the \emph{continuous} setting, the valuations are positive real numbers, which are without loss of generality normalized to sum to $1$. These normalizations make procedures invariant to any rescaling of the bids~\cite{KKP14a,BFM+12a}.

\subsection{The Adjusted Winner Procedure}

%\paragraph{The Adjusted Winner procedure}
The Adjusted Winner procedure works as follows. Alice and Bob are asked by a mediator to state their valuations $\mathbf{a}$ and $\mathbf{b}$, after which the next two phases are executed.

\begin{quote}
\textbf{\emph{Phase 1:}} For every item $i$, if $a_i > b_i$ then give the item to Alice; otherwise give it to Bob. The resulting allocation is $(W_A, W_B)$ and without loss of generality,
$u_{\vec{a}}(W_A) \geq u_{\vec{b}}(W_B)$.

\textbf{\emph{Phase 2:}} 
Order the items won by Alice increasingly by the \emph{ratio} $a_i/b_i$: $\frac{a_{k_1}}{b_{k_1}} \leq \ldots \leq \frac{a_{k_r}}{b_{k_r}}$. 
From left to right, continuously transfer fractions of items from Alice to Bob, until an allocation $(W_A', W_B')$ where both players have the same utility is produced:
$u_{\vec{a}}(W_A') = u_{\vec{b}}(W_B')$.
\end{quote}

Let $AW(\vec{a}, \vec{b})$ denote the allocation produced by Adjusted Winner on inputs $(\vec{a}, \vec{b})$, where $AW_A(\vec{a}, \vec{b})$ and $AW_B(\vec{a}, \vec{b})$
are the bundles received by Alice and Bob.
Note that the procedure is defined for strictly positive valuations, so the ratios are finite and strictly positive numbers.
Examples can be found on the Adjusted Winner website\footnote{http://www.nyu.edu/projects/adjustedwinner/.} as well as in~\cite{BT96a}.

Adjusted Winner produces allocations that are \emph{envy-free}, \emph{equitable}, \emph{Pareto optimal}, and \emph{minimally fractional}.
An allocation $W$ is said to be Pareto optimal if there is no other allocation that strictly improves one player's utility without degrading the other player.
%such that $u_i(W_i') \geq u_i(W_i)$ for each $i \in \mathbf{a,b}$ and $u_i(W_i') \geq u_i(W_i)$ for at least one of them.
Allocation $W$ is equitable if the utilities of the players are equal: $u_{\vec{a}}(W_A) = u_{\vec{b}}(W_B)$, envy-free if no player would prefer the other player's bundle, 
and minimally fractional if at most one item is split.

Envy-freeness of the procedure implies \emph{proportionality}, where an allocation is 
proportional if each player receives a bundle worth at least half of its utility for all the items.
A procedure is called envy-free if it always outputs an envy-free allocation (similarly for the other properties).

\section{Characterizations}\label{sec:characterizations}

In this section, we provide two characterizations of Adjusted Winner\footnote{
The results here refer to the case when the players report their true valuations to the mediator. We discuss the strategic aspects of the procedure in Section \ref{sec:strategic}.}
for both the discrete and continuous variants. We
begin with
a different interpretation of the procedure that is useful for analyzing its properties.

An allocation is \emph{ordered} if it can be produced by sorting the items in decreasing order of the valuation ratios $a_i/b_i$ and placing a boundary line somewhere (possibly splitting an item), such that Alice gets 
the entire bundle to the left of the line and Bob gets the remainder:
%%%over the ordered items such that each player gets the items for which they have relatively higher reports. 
$$\underbrace{\frac{a_{k_1}}{b_{k_2}}\geq \frac{a_{k_2}}{b_{k_2}}\geq \cdots \geq\frac{a_{k_i}}{b_{k_i}}\geq}_{\text{Alice's allocation}}  \bigg|\geq \underbrace{\frac{a_{k_{i+1}}}{b_{k_{i+1}}}\geq\cdots  \geq\frac{a_{k_m}}{b_{k_m}}}_{\text{Bob's allocation}}$$
The placement of the boundary line could lead either to an integral or a minimally fractional allocation. Note that the allocation that gives all the items to Alice is also ordered (but admittedly unfair).

It is clear to see that Adjusted Winner produces an ordered allocation (using some tie-breaking rule for items with equal ratios) with the property that the boundary line is appropriately placed to guarantee equitability. This is the way we will be interpreting the procedure for the remainder of the paper. We start by characterizing Pareto optimal allocations.
%
%
%first produces an ordering of ratios as the one above (using some tie-breaking rule) and then appropriately placing the boundary line in order to guarantee equitability, possibly fractionally assigning a single item. 
%We will use this interpretation throughout the paper.

%%% SIMINA:: Stating explicitely so that the readers get a reminder of how awesome the results are ;) --->> Since there is no ambiguity on the solution concept, we will will omit the term "pure" when referring to equilibria. 

% Note that AW results in a unique allocation if each $x_i/y_i$ is different for each item $i$. In case the ratio $x_i/y_i=x_j/y_j$ for items $i \neq j$, then the AW rule could return different possible allocations. Even though AW could result in different allocations, because of AW being pareto optimal and equitable the number of points received by the two players is the same in all of the possible allocations.

\begin{lemma}\label{lemma:cond-iff-notPO}
For any valuations $(\vec{a}, \vec{b})$ and any tie-breaking rule, an allocation $W$
is \emph{not} Pareto optimal if and only if there exist items $i$ and $j$ such that Alice gets a non-zero fraction (possibly whole) of $j$, 
Bob gets a non-zero fraction (possibly whole) of $i$, and $a_i b_j > a_j b_i$.
\end{lemma}
\begin{proof}
($\impliedby$) If such items $i,j$ exist, then consider the exchange in which Bob gives $\lambda_i>0$ of item $i$ to Alice and Alice gives $\lambda_j>0$ of item $j$ to Bob, where:
$$\frac{b_i}{b_j}\lambda_i <\lambda_j < \frac{a_i}{a_j}\lambda_i$$
Since $a_i / a_j > b_i/b_j$, such $\lambda_i$ and $\lambda_j$ do exist. 
%Now let us look at the change in utility.
Then Alice's net change in utility is:
$$
a_i\lambda_i - a_j\lambda_j > a_i\lambda_i - a_j\frac{a_i}{a_j}\lambda_i = 0,
$$
while Bob's net change is: %$y_j\lambda_j-x_y\lambda_i$. 	
$$
b_j\lambda_j-b_i\lambda_i > b_j\lambda_j-b_i(\lambda_j\frac{b_j}{b_i}) > b_j\lambda_j-{b_j}\lambda_j = 0.
$$
Thus the allocation is not Pareto optimal. \\

($\implies$) If the allocation $W$ is not Pareto optimal, then Alice and Bob can exchange positive fractions of items to get a Pareto improvement. 

Consider such an exchange and let $S_{A}$ be the set of items for which positive fractions are given by Alice to Bob. Let $S_{B}$ be defined similarly for Bob. Without loss of generality, $S_{A}$ and $S_{B}$ are disjoint; otherwise we could just consider the net transfer of any items that are in both $S_{A}$ and $S_{B}$. 
Let $j \in S_{A}$ be the item with the lowest ratio $a_j/b_j$, and $i \in S_{B}$ with the highest ratio $a_i/b_i$. 

If $a_i b_j> a_j b_i$ then we are done. Otherwise, assume by contradiction that for each item $k \in S_{A}$ and $l \in S_{B}$ it holds $a_k b_l\geq a_l b_k$. 
Then $a_k / b_k \geq a_l / b_l$; but then any Pareto improving exchange involving the transfer of items from $S_{A}$ and $S_{B}$ is only possible if at least one player gets a larger fraction of items without the other player getting a smaller fraction, which is impossible.
\end{proof}

By Lemma \ref{lemma:cond-iff-notPO}, a Pareto optimal allocation can be obtained by sorting the items by the ratios of the valuations and drawing a boundary line somewhere.
No matter where the boundary line is, the allocation is Pareto optimal (even if not equitable); thus an allocation is Pareto optimal and splits at most one item if and only if it is ordered. From this we obtain our first characterization.

\begin{comment}
%%%SIMINA:: I AM STATING THIS LEMMA INTEGRATED IN THE TEXT ABOVE BECAUSE OF LACK OF SPACE
\begin{lemma}\label{lemma}
An allocation is Pareto optimal and divides at most one item partially if and only if it is ordered.
\end{lemma}
\end{comment}

%We can now state our first characterization.
\begin{theorem}\label{th:charac1}
Adjusted Winner is the only Pareto optimal, equitable, and minimally fractional procedure. Any ordered equitable allocation can be produced by Adjusted Winner under some tie-breaking rule.
\end{theorem}
\begin{comment}
%%%% SIMINA: ADD BACK IF ENOUGH SPACE %%%%
\begin{proof}
First notice that any ordered and equitable allocation is the outcome of Adjusted Winner with some tie-breaking rule. 
Now consider any equitable Pareto optimal allocation that has at most one item divided partially. By the previous observations, such an allocation must be ordered; since it is equitable it must then be the outcome of Adjusted Winner 
with some tie-breaking rule.
\end{proof}
\end{comment}

%%\soren{should argue that AW can produce any ordered equitable allocation.}

Note that both Pareto optimality and equitability are necessary for the characterization. By restricting to Pareto optimal allocations only, then even the allocation that gives all the items to one player is Pareto optimal, while by restricting to equitable allocations only, even an allocation that throws away all the items is equitable. Similarly when the players have identical utilities for some items, then there exist Pareto optimal and equitable allocations that split more than one item. For example, if the two players have identical utilities over all items, then the allocation that gives half of each item to each player is equitable and Pareto optimal. However, in the case that the valuation are such that $a_i/b_i\neq a_j/b_j$ for all items $i \neq j$, 
then Adjusted Winner is exactly characterized by Pareto optimality and equitability.

We say that an allocation is \emph{maxmin} if it maximizes the minimum utility over both players.

\begin{theorem}\label{th:characb}
If the valuations satisfy $a_i/b_i\neq a_j/b_j$ for all items $i \neq j$, then the only Pareto optimal and equitable allocation is the result of Adjusted Winner.
\end{theorem}
\begin{proof}
Recall first that an allocation is maxmin if it maximizes the minimum utility of the players. Notice that AW achieves the same level of utility for the players.
Now
assume there exists an allocation $(\alpha,\beta)$ that is Pareto optimal and equitable, but not a result of AW. Then the allocation is not ordered and there exist at least two items $i$ and $j$ such that both players get a fraction of them. This contradicts Lemma \ref{lemma:cond-iff-notPO}, and so $(\alpha,\beta)$ does not exist.
\end{proof}

From Lemma 3.3~\cite{DaMo07a}, an allocation is maxmin if and only if it is Pareto optimal and equitable.
Together with Theorem~\ref{th:charac1}, this leads to another characterization.

\begin{theorem}
Adjusted Winner is equivalent to the procedure that always outputs a maxmin and minimally fractional allocation.
\end{theorem}

%Next, we examine the existence of pure Nash equilibria of Adjusted Winner.

%Having investigated the characterizations of the protocol, we now proceed to analyze the outcomes attainable when the players 
%are strategic and may use strategies different from their true valuations. We will return to the characterization given by Theorem (\simina{State NUMBER???})
%when analyzing the efficiency of its equilibria.

\begin{comment}
%%% SIMINA:: LEAVE THIS FOR A FULL VERSION OF THE PAPER; IT WOULD BE USEFUL PERHAPS TO INCLUDE AN EXAMPLE WITH A MANIPULATION?
if players are allowed to bid arbitrary units of points, then each player has an incentive to bid lower total number of points. Consider the allocation when both players bid the same number of total points. Now if Alice bids a lower number of total points, then the original allocation does not give equal utility to the players according to the new bids. In order to the equalize the utility achieved by both the players, the boundary of the original allocation is moved in favour of Alice. Hence Alice gets a more preferred allocation by simply scaling down his bids.
\end{comment}

\section{Equilibrium Existence}\label{sec:strategic}

In this section, we study Adjusted Winner when the players are \emph{strategic}, that is, their reported valuations are not necessarily the same as their actual valuations. Let $\mathbf{x}=(x_1,x_2,\ldots,x_m)$ and $\mathbf{y}=(y_1,y_2,\ldots,x_m)$ be the \emph{strategies} (i.e. declared valuations) of Alice and Bob respectively. Call $(\mathbf{x,y})$ a \emph{strategy profile}. We will refer to $\mathbf{a}$ and $\mathbf{b}$ as the \emph{true values} of Alice and Bob. Note that since strategies are reported valuations they are positive numbers that sum to $1$.

Since the input to Adjusted Winner is now a strategy profile $(\mathbf{x,y})$ instead of $(\mathbf{a,b})$, this means that the properties of the procedure are only guaranteed to hold with respect to the \emph{declared} valuations, and not necessarily the true ones\footnote{We will show that in the equilibrium, the procedure guarantees some of the properties with respect to the true values as well.}.
%
%all the results from Section \ref{sec:characterizations} can be rephrased in terms of the strategies.

A strategy profile $\mathbf{(x,y)}$ is an $\epsilon$-\emph{Nash equilibrium} if no player can increase its utility by more than $\epsilon$ by deviating to a different (pure) strategy. 
For $\epsilon=0$, we obtain a \emph{pure Nash equilibrium}. 

%The main result of this section is that $\epsilon$-Nash equilibria always exist. Furthermore, using an appropriate rule for settling ties between items with equal ratios $x_i/y_i$, the procedure also has exact pure Nash equilibria. We start our investigations from simple tie-breaking rules. 

The main result of this section is that Adjusted Winner is only guaranteed to have $\epsilon$-Nash equilibria when $\epsilon>0$ using standard tie-breaking. For the discrete case, this is achieved by the center setting the number of points or equivalently the denominator large enough. Furthermore, we prove that when using an appropriate rule for settling ties between items with equal ratios $x_i/y_i$, the procedure does admit exact pure Nash equilibria. We start our investigations from the standard tie-breaking rules. 

\subsection{Lexicographic Tie-Breaking}
%
% \haris{Have we mentioned somewhere that using less points helps a player?}

%%%SIMINA:: Yep, these are going in the intro to the respective sections:: \haris{Prominently define `uninformed' by italicizing it the first time it is defined}
%%%\haris{Spell out what the discrete and continious variants actually are}

The classical formulation of Adjusted Winner resolves ties in an arbitrary deterministic way, for example by ordering the items lexicographically, such that items with lower indices come first. 

\subsubsection{Continuous Strategies}
First, we consider the case of continuous strategies. We start with the following theorem.

\begin{theorem}\label{cont:no-PNE}
Adjusted Winner with continuous strategies is not guaranteed to have pure Nash equilibria.
\end{theorem}
\begin{proof}
Take an instance with two items and valuations $(\vec{a}, \vec{b})$, where $b_1 >  a_1 > a_2 > b_2 > 0$.
Assume by contradiction there is a pure Nash equilibrium at strategies $(\vec{x}, \vec{y})$, where $\vec{x} = (x, 1 - x)$ and $\vec{y} = (y, 1-y)$.
%We have a few cases.
We study a few cases and show the players can always improve. \\

\textbf{\emph{Case 1}}: ($x \neq y$). Without loss of generality $x>y$ (the case $x<y$ is similar). Then there exists $ \delta \in \mathbb{R}$ with $x-\delta > y \Rightarrow 1-x +\delta < 1 - y$, and Alice can improve by playing $\mathbf{x'}=(x-\delta, 1-x+\delta)$, as the boundary line moves to the left of its former position. \\

\textbf{\emph{Case 2}}: ($x=y<1/2)$. Here both players report higher values on the item they like less; Alice's allocation is $(1,\lambda)$ while Bob's is $(0,1-\lambda$), for some $\lambda \in (0,1)$. Then $\exists \; \delta \in \mathbb{R}$ with $x+\delta < 1/2$. By playing $\mathbf{y'}=(x+\delta, 1-x-\delta)$, Bob gets $(1,1-\lambda')$, for some $\lambda' \in (0,1)$. This is a strict improvement since $a_1 > a_2$. \\

%Consider such a $\delta$ and strategy $\mathbf{y'}=(x+\delta, 1-x-\delta)$ for Bob. Under $\mathbf{y'}$, Bob's allocation is $(1,1-\lambda')$ for some %$\lambda' \in (0,1)$, which is a strict improvement since $a_1 > a_2$, contradiction.

\textbf{\emph{Case 3}}: ($x=y>1/2$). Both players report higher values on the item they like more. Bob gets $(1-\frac{1}{2x},1)$ and Alice gets $(\frac{1}{2x},0)$, with utilities:
$$u_{\vec{a}}(AW(\vec{x},\vec{y})) = \frac{a_1}{2x}$$ 
and 
$$u_{\vec{b}}(AW(\vec{x},\vec{y}))  = \left(1-\frac{1}{2x}\right)b_1+b_2.$$ 

Let $\delta \in (0,\min(1-x,2x-1))$ such that:
$$\delta < \max\left\{\frac{4x(x-a_1)}{2x-a_1},\frac{4x(b_1-x)}{2x-b_1}\right\}.$$ Observe that since $b_1 > a_1$ and $2x-a_1$ and $2x-b_1$ are positive, at least one of $x-a_1$ and $b_1-x$ is strictly positive and by continuity of the strategy space, such a $\delta$ exists. 
%
% means that Alice's utility is $\frac{1}{2x}a_1$ and Bob's utility is $\left(1-\frac{1}{2x}\right)b_1+b_2$. 
Now consider alternative profiles $\mathbf{(x',y)} = ((x-\delta, 1-x+\delta), (x,1-x))$ and $\mathbf{(x,y')}=((x,1-x),(x+\delta,1-x-\delta))$. Since $\delta < 2x-1$, the first item is still the item that gets split in the new profile.
Using the identities $a_1+a_2=b_1+b_2=1$ and the assumption that $\mathbf{(x,y)}$ is a pure Nash equilibrium, we have that
\begin{equation*} 
\left\{\begin{array}{l}
  a_1\left(1-\frac{1}{2x}-\frac{1}{2x-\delta}\right) + a_2 \leq 0 \implies \delta \geq \frac{4x(x-a_1)}{2x-a_1}\\
  b_1 \left(1-\frac{1}{2x}-\frac{1}{2x+\delta}\right) + b_2 \geq 0 \implies \delta \geq  \frac{4x(b_1-x)}{2x-b_1}
  \end{array}
\right.
\end{equation*}

We obtain a contradiction, so this case cannot occur. \\
 
\textbf{\emph{Case 4}}: ($x=y=1/2$). Alice and Bob get allocations $(1,0)$ and $(0,1)$, respectively.
%In this case Bob's allocation is $(0,1)$. 
Let $0 < \delta < \frac{(b_1-b_2)}{b_2}$ and consider the strategy $\mathbf{y'}=(x+\delta, 1-x-\delta)$ of Bob. Using $\vec{y}'$, Bob gets the allocation $(\frac{1}{\delta+1},0)$, which is better than $(0,1)$. Since $b_1>b_2$, such $\delta$ exists. \\

As none of the cases $(1)-(4)$ are stable, the procedure has no pure Nash equilibrium.
\end{proof}

However, we show that Adjusted Winner admits approximate Nash equilibria.

%\textbf{NOTE:} Without loss of generality, in the following (and probably also previously) we can assume that each player has a budget of $1$ dollar that can be split continuously.

\begin{theorem} \label{thm:epsilonNash}
Each instance of Adjusted Winner with continuous strategies has an $\epsilon$-Nash equilibrium, for every $\epsilon > 0$.
\end{theorem}
\begin{proof}
Let $(\vec{a}, \vec{b})$ be any instance. We show there exists an $\epsilon$-Nash equilibrium in which
Alice plays her true 
valuations and Bob plays a small perturbation of Alice's valuations.
More formally, we show there exist $\epsilon_1, \ldots, \epsilon_m$, 
such that an $\epsilon$-equilibrium is obtained when Alice plays $\vec{a} = (a_1, \ldots, a_m)$ and Bob plays
$\tilde{\vec{a}} = (\tilde{a}_1, \ldots, \tilde{a}_m)$, where $\tilde{a}_i = a_i + \epsilon_i$ for each item $i \in [m]$ and $\sum_{i=1}^{m} \epsilon_i = 0$. The theorem will follow from the next two lemmas.
\end{proof}

\begin{comment}
Recall that for any strategies $(\vec{x}, \vec{y})$ of Alice and Bob, respectively, the outcome of the procedure is computed by
\begin{itemize} 
\item sorting the items 
in decreasing order of the $\frac{x}{y}$ ratios: $\frac{x_{\pi_1}}{y_{\pi_1}} \geq \ldots \geq \frac{x_{\pi_m}}{y_{\pi_m}}$, 
with ties broken lexicographically and $\pi = (\pi_1, \ldots, \pi_m)$ 
a permutation of $\{1, \ldots, m\}$
\item finding the uniquely defined index $l$ and $\lambda$ such that: 
$x_{\pi_1} + \ldots x_{\pi_{l-1}} + \lambda x_{\pi_l} = (1 - \lambda) y_{\pi_l} + y_{\pi_{l+1}} + \ldots + y_{\pi_m}$.
\end{itemize}
\end{comment}

%Let $V_a$ and $V_b$ denote the valuation functions induced by preference vectors $(\vec{a}, \vec{b})$,
%and
%$X = (X_a, X_b)$ the resulting allocation given preference vectors $(\vec{a}, \vec{b})$.

\begin{lemma} \label{lem:Aliceepsilon}
For 
any pair of strategies $(\vec{a}, \tilde{\vec{a}})$, where $|a_i - \tilde{a}_i| < \epsilon/m$ for all $i \in [m]$,
Alice's strategy is an $\epsilon$-best response.
\end{lemma}
\begin{proof}
Since the procedure is envy-free, Alice gets at least half of the total value by being truthful regardless of Bob's strategy, and so $u_{\vec{a}}(AW_A(\vec{a},\vec{\tilde{a}})) \geq 1/2$.
%%%%\mathcal{W}_{\vec{a},\vec{\tilde{a}}}^{A})) \geq 1/2$.
The allocation must also be envy-free according to Bob's declared valuation profile $\tilde{\vec{a}}$, 
and so $u_{\tilde{\vec{a}}}(AW_B(\vec{a}, \tilde{\vec{a}})) \geq 1/2$.

Since strategies $\vec{a}$ and $\tilde{\vec{a}}$ are $\epsilon$-close, that is $\sum_i |a_i - \tilde{a}_i| < \epsilon$, then their evaluations of the same allocation, namely $AW_B(\vec{a}, \tilde{\vec{a}})$, are also close:
\begin{eqnarray*}
u_{\vec{a}}(AW_B(\vec{a}, \tilde{\vec{a}}))
\geq u_{\tilde{\vec{a}}}(AW_B(\vec{a}, \tilde{\vec{a}})) - \epsilon \geq 1/2 - \epsilon
\end{eqnarray*}
It follows that $1/2 \leq u_{\vec{a}}(AW_A(\vec{a}, \tilde{\vec{a}}))
\leq 1/2 + \epsilon$.
Moreover, Alice cannot use some other strategy $\vec{a}'$ to force an allocation that gives her more than $1/2 + \epsilon$;
otherwise, Bob's utility as measured by $\tilde{\vec{a}}$ under strategy profiles $(\vec{a}', \tilde{\vec{a}})$ would be strictly less than $1/2 - \epsilon$,
contradicting the envy-freeness of the procedure.

Thus when Bob's strategy is $\epsilon$-close to Alice's truthful strategy $\vec{a}$, Alice has an $\epsilon$-best response at her truthful strategy $\vec{a}$, which completes the proof of the lemma.
\end{proof}

\begin{lemma} \label{lem:Bobepsilon}
When Alice plays $\vec{a}$, Bob has an $\epsilon$-best response that is $\epsilon$-close to Alice's strategy.
\end{lemma}
\begin{proof}
%Let $\vec{z}$ be an arbitrary $\epsilon$-best response of Bob to Alice's strategy $\vec{a}$. Assume that $\vec{z}$ is \emph{far} from Alice's strategy; that is,
%$||\vec{z} - \vec{a}|| > \epsilon$.
%
Let $\pi = (\pi_1, \ldots, \pi_m)$ be a fixed permutation of the items. %by strategy profiles $(\vec{a}, \vec{z})$. 
Then there exist uniquely defined index $l \in \{1, \ldots, m\}$ and $\lambda \in [0,1)$ such that 
\begin{equation} \label{eq:idealoutcomepi}
a_{\pi_1} + \ldots a_{\pi_{l-1}} + \lambda a_{\pi_l} = \frac{1}{2} = (1 - \lambda) a_{\pi_l} + a_{\pi_{l+1}} + \ldots + a_{\pi_m}
\end{equation}
Note that Adjusted Winner uses lexicographic tie breaking to sort the items when there exist equal ratios $x_i/y_i = x_j/y_j$, for some $i\neq j$. Thus
the order $\pi$ may never appear in an outcome of the procedure when the players use the same strategies.

However, we show that Bob can approximate the outcome of Equation (\ref{eq:idealoutcomepi}) arbitrarily well. We have two cases: \\

\textbf{\emph{Case 1}}: $\lambda \in (0,1)$.
Then there exist $\epsilon_1, \ldots, \epsilon_m$ such that the following conditions hold: 
\begin{description}
\item[$(i)$] $|\epsilon_j| < \min\left( \frac{\epsilon}{m}, \frac{2\lambda a_{\pi_l}}{m}\right)$, for all $j \in [m]$,
\item[$(ii)$] the items are strictly ordered by $\pi$: $\frac{a_{\pi_1}}{a_{\pi_1} + \epsilon_{\pi_1}} > \ldots > \frac{a_{\pi_m}}{a_{\pi_m} + \epsilon_{\pi_m}}$, 
\item[$(iii)$] $\sum_{j=1}^{m} \epsilon_j = 0$, and 
\item[$(iv)$] it's still item $\pi_l$ that gets split, in a fraction $\delta \in (0,1)$ close to $\lambda$; that is,
$|\lambda - \delta| < \frac{\epsilon}{b_{\pi_l}}$. 
\end{description}

Informally, Bob plays a perturbation of Alice's truthful strategy inducing ordering $\pi$ on the items (with no ties) and splits item $\pi_l$
in a fraction close to $\lambda$. \\

\textbf{\emph{Case 2}}: $\lambda = 0$. 
Again, there exist $\epsilon_1, \ldots, \epsilon_m$ such that the following conditions are met:
\begin{description}
\item[$(i)$] $\epsilon_j < \min\left(\frac{\epsilon}{m}, \frac{a_{\pi_l}}{m}\right)$ for all $j \in [m]$, 
\item[$(ii)$] the item order is $\pi$:
$\frac{a_{\pi_1}}{a_{\pi_1} + \epsilon_{\pi_1}} > \ldots > \frac{a_{\pi_m}}{a_{\pi_m} + \epsilon_{\pi_m}}$,
\item[$(iii)$] $\sum_{j=1}^{m} \epsilon_j = 0$, and 
\item[$(iv)$] item $\pi_l$ is split in a ratio $\delta$ close to zero: $|\delta| < \frac{\epsilon}{b_{\pi_l}}$.
\end{description}

Thus Bob can approximate the outcome of Equation (\ref{eq:idealoutcomepi}). 

Now consider any $\epsilon$-best response $\vec{y}$ of Bob; this induces some permutation of the items according to the ratios. If $\vec{y}$ is $\epsilon$-close to the strategy of Alice we are done. Otherwise, Bob could change his strategy to be $\epsilon$-close to the strategy of Alice while inducing the same permutation. This will only improve his utility as the boundary line moves to the left.
\end{proof}

It can be observed that there is at least one other $\epsilon$-Nash equilibrium, at strategies $(\vec{b}, \tilde{\vec{b}})$, where $\tilde{\vec{b}}$ is a perturbation of Bob's truthful profile.

\subsubsection{Discrete Strategies}

Even though the continuous procedure is not guaranteed to have pure Nash equilibria, this does not imply that the discrete variant should also fail to have pure Nash equilibria.
However we do find that this is indeed the case. 
%%%\Frowny{}.

\begin{theorem}\label{disc:no-PNE}
Adjusted Winner with discrete strategies is not guaranteed to have pure Nash equilibria.
\end{theorem}
\begin{proof}
	Consider a game with $4$ items and $7$ points, where Alice and Bob have valuations $(1, 1, 2, 3)$ and $(2,3,1,1)$, respectively.
	This game does not admit a pure Nash equilibrium; this fact can be verified with a program that checks all possible configurations.
%%%%	\haris{Can we argue for this analytically?-->> SIMINA:: SORRY, I DON'T SEE HOW TO DO THAT}
	\end{proof}

% \begin{example}[Instance with no pure Nash equilibrium]
% Consider a game with $4$ items and $7$ points for each player, where Alice and Bob have valuation profiles $[1, 1, 2, 3]$ and $[2,3,1,1]$, respectively.
% This game does not admit a pure Nash equilibrium; this fact can be verified with a computer program that checks all the possible configurations.
% \haris{Can we argue for this analytically?r}
% \end{example}

\begin{comment}
In the case of two items however, an exact equilibrium always exists.

\begin{theorem}\label{th:pne-2items}
In the case of two items, Adjusted Winner always has a pure Nash equilibrium.
\end{theorem}  
\end{comment}

Our next theorem shows that an $\epsilon$-Nash equilibrium always exists in the discrete case if the number of points is set adequately, such that the players can approximately represent their true valuations.

\begin{theorem} \label{thm:epsilonNashDiscrete}
For any profile $(\vec{a}, \vec{b})$ and any $\epsilon > 0$, there exists $P'$ such that the procedure has an $\epsilon$-Nash equilibrium when the players are given $P'$ points.
\end{theorem}
\begin{proof}
Let $\epsilon > 0$, and consider any profile $(\vec{a}, \vec{b})$ with denominator $P$. Then if we interpret $(\vec{a}, \vec{b})$ as a profile for the continuous setting, we get a $\epsilon/2$-Nash equilibrium $(\vec{a}, \vec{\tilde{a}})$ from Theorem~\ref{thm:epsilonNash}, where $\tilde{a}_j = a_j + \epsilon_j$, for all $j \in [m]$. 

Recall that $a_j, b_j \in \mathbb{Q}$; where $a_j = \frac{s_j}{P}$ and $b_j = \frac{t_j}{P}$, for some $s_j, t_j, \in \mathbb{N}$. We can find a rational number
$\epsilon'_j = \frac{q_j}{r_j}$ (with $q_j, r_j \in \mathbb{N}$) that approximates $\epsilon_j$ within $\frac{\epsilon}{2m}$ for each $j \in [m]$,
and such that the ordering of the items induced by the ratios $\frac{a_j}{a_j + \epsilon_j}$ is the same as the one given by $\frac{a_j}{a_j + \epsilon'_j}$.
Define $\vec{\tilde{a}}'$ such that $\tilde{a}'_j = a_j + \epsilon'_j$.

It follows that $(\vec{a}, \vec{\tilde{a}}')$ is an $\epsilon$-Nash equilibrium with $a_j, \tilde{a}'_j \in \mathbb{Q}$, for all $j \in [m]$. Thus
whenever the players have a denominator of $P' = P \cdot \prod_{j=1}^{m} r_j$, the strategy profiles $(\vec{a}, \vec{\tilde{a}}')$ can be represented in the discrete procedure, so by giving $P'$ points to the players, there exists an $\epsilon$-Nash equilibrium.
\end{proof}

%%% SIMINA ::: NO, ACTUALLY I DON'T THINK SO%%%
%\begin{conjecture}
%Don't we get an exact equilibrium actually for the discrete case, by setting the number of coins large enough?
%\end{conjecture}

\subsection{Informed Tie-Breaking}

If the tie-breaking rule is not independent of the valuations, then both the discrete and continuous variants of Adjusted Winner have exact pure Nash equilibria. The deterministic tie-breaking rule under which this is possible 
is the one in which a fixed player (e.g. Bob), is allowed to resolve ties by sorting them in the best possible order for him. That is, Bob evaluates all ways of sorting the items with ties and picks the ordering that maximizes his utility, according to his true valuation function. If there are multiple such orderings, Bob can without loss of generality select any of them.

\begin{comment}
Bob can compute the optimal order as outlined in the next definition.

\begin{definition}[Informed Tie-Breaking]
Let there be a fixed player, for example Bob. Given any strategies $(\vec{x}, \vec{y})$,
for each permutation $\pi$, let $l_{\pi} \in [m]$ and $\lambda_{\pi} \in [0, 1)$ be the uniquely defined item and fraction for which:
$$x_{\pi_1} + \ldots x_{\pi_{l-1}} + \lambda x_{\pi_l} 
= (1 - \lambda) y_{\pi_{l}} + y_{\pi_{l+1}} + \ldots + y_{\pi_m}$$
Let $\pi^{*}$
be an optimal permutation with respect to $(\vec{x}, \vec{y})$, namely
$\pi^{*} \in \arg \max_{\pi} (1 - \lambda)y_{\pi_l} + y_{\pi_{l+1}} + \ldots + y_{\pi_m}$.
Then under \emph{informed tie-breaking}, the procedure resolves ties in the order given by $\pi^{*}$.
\end{definition}

Note that there might be more than one choice of $\pi^*$ and Bob picks any fixed one.  
\end{comment}
Now we can state the equilibrium existence theorems.

\begin{theorem} \label{thm:continuousInformedExactNash}
Adjusted Winner with continuous strategies and informed tie-breaking is guaranteed to have a pure Nash equilibrium.
\end{theorem}
\begin{proof}
We show that the profile $(\vec{a}, \vec{a})$ is an exact equilibrium. By envy-freeness of the procedure, Alice gets at least half of the points at this strategy profile. Moreover, she cannot get strictly above 
half, since that would violate envy-freeness from the point of view of Bob's declared valuation, which is also $\vec{a}$. Thus Alice's strategy is a best response.
As argued in Theorem~\ref{thm:epsilonNash} and ~\ref{thm:epsilonNashDiscrete}, there exists an optimal permutation $\pi^*$ such that by playing $\vec{a}$ and sorting the items in the order $\pi^*$, Bob can obtain the best possible 
utility (and as mentioned in Lemma~\ref{lem:Bobepsilon}, this value is achievable at these strategies).
\end{proof}

Similarly, it can be shown that the strategy profile $(\vec{a}, \vec{a})$ is a pure Nash equilibrium in the discrete procedure.
\begin{theorem} \label{thm:discreteInformedExactNash}
Adjusted Winner with discrete strategies and informed tie-breaking is guaranteed to have a pure Nash equilibrium.
\end{theorem}
%%%%SIMINA::: ADD BACK IF ENOUGH SPACE
\begin{comment}
\begin{proof}
Consider the strategy profile $(\vec{a}, \vec{a})$. From Theorem~\ref{thm:continuousInformedExactNash}, this is a Nash equilibrium in the continuous case. Since the strategy space in the discrete
procedure is more restricted, there are no improving deviations here either, and so the theorem follows. 
\end{proof}
\end{comment}

\section{Efficiency and Fairness of Equilibria}

Having examined the existence of pure Nash equilibria in Adjusted Winner, 
we now study their fairness and efficiency. For fairness, we observe that following.
\begin{theorem}\label{thm:ne-ef}
All the pure Nash equilibria of Adjusted Winner are envy-free with respect to true valuations of the players.
\end{theorem}
\begin{proof}
Each player is guaranteed at least $50\%$ of the maximum utility by playing truthfully, regardless of what the other player does.
Since a player always has truthful reporting as a possible strategy, it must be the case that any equilibrium
outcome guarantees $50\%$ as well. Since the total utility is $100\%$, the allocation is envy-free.
\end{proof}

For efficiency, we use the well known measure of the \emph{Price of Anarchy}~\cite{koutsoupias1999worst, AGT_book}.

The \emph{social welfare} of an allocation $W$ is defined as the sum of the players' utilities: $$SW(W) = u_A(W_A) + u_B(W_B).$$

Then the Price of Anarchy is defined as the social welfare achieved in the outcome of Adjusted Winner (when the players are not strategic) over the social welfare achieved in the worst pure Nash equilibrium of the procedure and measures the deterioration of the welfare due to the strategic behaviour of the players.
% \aris{The optimal allocation can be outputted by Adjusted Winner for some input strategies $(\mathbf{x,y})$. Additionally, every minimally fractional allocation can be outputted for some strategies.}
Our main findings are that when the procedure is equipped with an informed tie-breaking rule $(i)$ all the pure Nash equilibria are Pareto optimal with respect to the true 
valuations and $(ii)$ the price of anarchy is constant; that is, each pure Nash
equilibrium achieves at least 75\% of the truthful social welfare.\footnote{Note that the \emph{optimal} welfare is not necessarily achievable in any outcome of the procedure; however, any Nash equilibrium of Adjusted Winner also attains a 50\% fraction of the optimal since the equilibria are envy-free with respect to the true valuations.}

We start with a lemma.

%In this section, we study the efficiency of the pure Nash equilibria of AW. We prove that when the procedure is equipped with an informed tie-breaking rule, all the pure Nash equilibria are %Pareto optimal with respect to the true values $\mathbf{a}$ and $\mathbf{b}$ of Alice and Bob respectively. In turn, this result has a direct implication on the Price of Anarchy of the %procedure.  

%\simina{What does $\tau$ mean? I thought we somehow defined $\tau$ uniquely?}

%\simina{Define at the beginning (and make sure this also holds in the existence sections) informed tie-breaking as Bob breaking ties in his favor, and choosing one possibility -- any among them - for further breaking ties}.

\begin{lemma}\label{lem:tie}
Let $\mathbf(\vec{x},\vec{x})$ be a pure Nash equilibrium of Adjusted Winner with informed tie-breaking and let $\pi^*$ be the permutation that Bob chooses. Then, among all possible permutations, $\pi^*$ maximizes Alice's utility. 
\end{lemma}

\begin{proof}
	Assume by contradiction that there exists a permutation $\pi$ that gives Alice a strictly larger utility; let $\alpha$ be her marginal increase from $\pi^*$ to $\pi$. As discussed in Section \ref{sec:strategic}, Alice can find appropriate constants $\epsilon_1, \ldots, \epsilon_m$ such that $AW(\mathbf{\vec{x}',\vec{x}})$ with $\mathbf{\vec{x}'}=(x_1+\epsilon_1,\ldots,x_m+\epsilon_m)$ orders the items by $\pi$ and the allocations $AW\mathbf{(\vec{x},\vec{x})}$ and $AW\mathbf{(\vec{x}',\vec{x})}$ differ only in the allocation of the split item by by $\delta$. Moreover, by continuity of the strategies, for each $\alpha$, there exist $\epsilon_i$'s such that $\delta$ is small enough for $AW\mathbf{(x',x)}$ to be better for Alice than $AW\mathbf{(x,x)}$.   
\end{proof}

Next we show that all equilibria are Pareto optimal.

\begin{theorem}\label{ne-po}
	All the pure Nash equilibria of Adjusted Winner with informed tie-breaking are Pareto optimal with respect to the true valuations $\mathbf{a}$ and $\mathbf{b}$. 
\end{theorem}

\begin{proof}
	Let $(\mathbf{x,x})$ be a pure Nash equilibrium of Adjusted Winner under informed tie-breaking and let $l$ be the item that gets split (if any, otherwise the item to the left of the boundary line). Order Alice's items decreasing order of ratios $a_i/x_i$ and Bob's items in increasing order of ratios $b_i/x_i$. Since $(\mathbf{x,x})$ is a pure Nash equilibrium, by Lemma \ref{lem:tie}, both players are getting their maximum utility over all possible tie-breaking orderings of items. This means that for every item $i \leq l$ and every item $j \geq l$ with $i \neq j$, it holds that
	$$\frac{a_j}{x_j} \geq \frac{a_i}{x_i} \ \ \text{and} \ \  \frac{b_i}{x_i} \geq \frac{b_j}{x_j}  \Rightarrow \frac{a_i}{x_i}\cdot\frac{b_j}{x_j} \leq \frac{a_j}{x_j}\cdot\frac{b_i}{x_i},$$
	which by Lemma 	\ref{lemma:cond-iff-notPO}, implies that $AW(\mathbf{x,x})$ is Pareto optimal.
\end{proof}

The Pareto optimality of a strategy profile has a direct implication on the social welfare achieved at that profile.

\begin{theorem}\label{poa}
	 The Price of Anarchy of Adjusted Winner is $4/3$.
\end{theorem}\begin{proof}
	Let $(\vec{x},\vec{y})$ be any  pure Nash equilibrium and let $V_A$ and $V_B$ be the utilities of Alice and Bob respectively from the outcome of Adjusted Winner on the truthful profile $(\mathbf{a,b})$, i.e. $V_A=u_A(AW(\mathbf{a,b}))$ and $V_B=u_B(AW(\mathbf{a,b}))$. Since $AW(\vec{x},\vec{y})$ is Pareto optimal by Theorem \ref{ne-po}, the allocation for at least one of the players, (e.g. Alice), is at least as good as that of the truthful outcome allocation. In other words, $u_A(AW(\vec{x},\vec{y})) \geq V_A$. On the other hand, since $AW(\vec{x},\vec{y})$ is envy-free by Theorem \ref{thm:ne-ef}, Bob's utility from $AW(\vec{x},\vec{y})$ is at least $1/2 \cdot V_B$. Overall, the social welfare of $AW(\vec{x},\vec{y})$ is at least $V_A+\frac{1}{2}V_B$. Since $V_A=V_B$ by equitability, the bound follows. %In the full version of the paper, we give an example for which the bound is tight. 
	
	The bound is (almost) tight, given by the following instance with two items. Let $\mathbf{a}=(1-\epsilon,\epsilon)$ and $\mathbf{b}=(\epsilon,1-\epsilon)$ and consider the strategy profile $\mathbf{x}=(\epsilon,1-\epsilon)$ and $\mathbf{y}=(\epsilon,1-\epsilon)$. It is not hard to see that $\mathbf{x,y}$ is a pure Nash equilibrium for Alice breaking ties. The social welfare of the truthful outcome is $2 - 2\epsilon$, in which each player receives their most preferred item. In the equilibrium allocation of Adjusted Winner, Alice wins the first item and the second item is split (almost) in half. The social welfare of the mechanism is $1+\frac{1}{2} + o(\epsilon)$ and the approximation ratio is (almost) $4/3$. As $\epsilon$ grows smaller, the ratio becomes closer to $4/3$.  
\end{proof}

\begin{remark}
Note that in fact it is possible that players have improved welfare at some Nash equilibria compared to the welfare at the truthful profile. To see this, consider an instance with $m=2$ items, where Alice has valuation profile is $(50,50)$, while Bob has valuation $(60,40)$ (without normalization).
The optimal welfare is achieved when Bob gets the first item and Alice gets the second one. However, the Adjusted Winner outcome has a welfare of $109.1$, since the protocol transfers a part of good $1$ to Alice in order to achieve equitability.
However, by providing a different input, of $(50,50)$, Bob can move the boundary line and ``restore'' the optimal welfare outcome; this is a Nash equilibrium when Bob breaks ties. 
\end{remark}
\section{Future Work}

According to Foley \cite{Foley67}, the quintessential characteristics of fairness are envy-freeness and Pareto optimality. We show that Adjusted Winner is guaranteed to have pure Nash equilibria, which satisfy both of these fairness notions. This attests to the usefulness and theoretical robustness of the procedure. A very interesting direction for future work is to study the \emph{imperfect information} setting, as the Nash equilibria studied here require the players to have full information of each other's preferences.
%\aris{Add observation about the allocation $(a,a)$ being the outcome where Alice gets the worst possible value and Bob the best possible}.
%\simina{TODO:: Discuss in detail the results for the two items case (with link to the appendix or the full version of the paper). 
%Also include all the ideas for future work, the generalization to multiple players, the connection with bargaining solutions (this appeared also in some references)}

%In this paper, we presented characterizations and results concerning the existence/non-existence of %equilibria of one of the most important fair division rules.  
%An interesting direction for future work is to consider the complexity of computing a best response as well checking the existence of pure Nash equilibria under Adjusted Winner. 
%There are various ways Adjusted Winner can be extended to multiple ways. Identifying the `right' generalization will be fruitful. Finally, it will be interesting to see how the arguments for existence and non-existence of pure Nash equilibria can be used for other mechanisms.
%\simina{Remove frowning face after it's been noticed by all the authors :-$)$}

\section{Acknowledgements}
Haris Aziz acknowledges support from NICTA, which is funded by the Australian Government through the Department of
Communications and the Australian Research Council through the ICT Centre
of Excellence Program.

Simina Br\^anzei, Aris Filos-Ratsikas, and S{\o}ren Kristoffer Stiil Frederiksen 
acknowledge support 
from the Danish National Research Foundation
and The National Science Foundation of China (under the grant 61361136003) for
the Sino-Danish Center for the Theory of Interactive Computation and from the Center for
Research in Foundations of Electronic Markets (CFEM), supported by the Danish
Strategic Research Council.

\bibliographystyle{plain}

\begin{thebibliography}{10}

\bibitem{BFM+12a}
S.~J. Brams, M.~Feldman, J.~Morgenstern, J.~K. Lai, and A.~D. Procaccia.
\newblock On maxsum fair cake divisions.
\newblock In {\em Proceedings of the Twenty-Sixth AAAI Conference on Artificial
  Intelligence}, pages 1285--1291. AAAI Press, 2012.

\bibitem{BT96a}
S.~J. Brams and A.~D. Taylor.
\newblock {\em Fair Division: From Cake-Cutting to Dispute Resolution}.
\newblock Cambridge University Press, 1996.

\bibitem{BT96b}
S.~J. Brams and A.~D. Taylor.
\newblock A procedure for divorce settlements.
\newblock {\em Issue Mediation Quarterly Mediation Quarterly}, 13(3):191--205,
  1996.

\bibitem{BrTa00a}
S.~J. Brams and A.~D. Taylor.
\newblock {\em The Win-Win Solution: Guaranteeing Fair Shares to Everybody}.
\newblock Norton, 2000.

\bibitem{BrTo96a}
S.~J. Brams and J.~M. Togman.
\newblock Camp david: Was the agreement fair?
\newblock {\em Conflict Management and Peace Science}, 15(1):99--112, 1996.

\bibitem{DaMo07a}
M.~Dall'Aglio and R.~Mosca.
\newblock How to allocate hard candies fairly.
\newblock {\em Mathematical Social Sciences}, 54(3):218----237, 2007.

\bibitem{daniel2005fair}
T.~E. Daniel and J.~E. Parco.
\newblock Fair, efficient and envy-free bargaining: An experimental test of the
  brams-taylor adjusted winner mechanism.
\newblock {\em Group Decision and Negotiation}, 14(3):241--264, 2005.

\bibitem{DeBr97a}
D.~B.~H. Denoon and S.~J. Brams.
\newblock Fair division: A new approach to the spratly islands controversy.
\newblock {\em International Relations}, 2(2):303--329, 1997.

\bibitem{Foley67}
D.~K. Foley.
\newblock Resource allocation and the public sector.
\newblock {\em Yale Econ Essays, Vol 7, No 1, pp 45-98, Spring 1967. 7 Fig, 13
  Ref.}, 1967.

\bibitem{KKP14a}
J.~Karp, A.~M. Kazachkov, and A.~D. Procaccia.
\newblock Envy-free division of sellable goods.
\newblock In {\em Proceedings of the Twenty-Eighth AAAI Conference on
  Artificial Intelligence}, pages 728--734. AAAI Press, 2014.

\bibitem{koutsoupias1999worst}
E.~Koutsoupias and C.~Papadimitriou.
\newblock Worst-case equilibria.
\newblock In {\em STACS 99}, pages 404--413. Springer, 1999.

\bibitem{Lax99a}
J.~R. Lax.
\newblock Fair division: A format for the debate on the format of debates.
\newblock {\em Political Science and Politics}, 32(1):45--52, 1999.

\bibitem{Levy99a}
G.~M. Levy.
\newblock Resolving real estate disputes.
\newblock {\em Real Estate Issues}, 1999.

\bibitem{Mada10a}
K.~Madani.
\newblock Game theory and water resources.
\newblock {\em Journal of Hydrology}, 381:225----238, 2010.

\bibitem{Mass00a}
T.~G. Massoud.
\newblock Fair division, adjusted winner procedure (aw), and the
  israeli-palestinian conflict.
\newblock {\em Journal of Conflict Resolution}, 44(3):333--358, 2000.

\bibitem{Nalebuff}
B.~Nalebuff.
\newblock Review of the win-win solution: Guaranteeing fair shares to
  everybody.
\newblock {\em Journal of Economic Literature}, 39:125--127, 2001.

\bibitem{AGT_book}
N.~Nisan, T.~Roughgarden, E.~Tardos, and V.~Vazirani.
\newblock {\em {Algorithmic Game Theory}}.
\newblock Cambridge University Press, (editors) 2007.

\bibitem{Rait00a}
M.~G. Raith.
\newblock Fair-negotiation procedures.
\newblock {\em Mathematical Social Sciences}, 39:303----322, 2000.

\bibitem{TaPa08a}
A.~D. Taylor and A.~M. Pacelli.
\newblock {\em Mathematics and Politics}.
\newblock Springer, 2008.

\end{thebibliography}

\end{document}